\documentclass[reqno,11pt]{amsart}
\newtheorem{definition}{Definition}
\usepackage{graphicx}
\usepackage{xcolor}
\usepackage{tikz}
\usepackage{tikz-cd}
\usepackage{caption}
\usepackage{booktabs}
\usepackage{adjustbox}
\usepackage{longtable}
\usepackage{array}
\usepackage{amssymb}

\usepackage{amsmath}
\newcommand{\be}{\begin{equation}}
\newcommand{\ee}{\end{equation}}

\DeclareMathSymbol{\Lambda}{\mathord}{operators}{"03}

\usepackage{enumerate}% http://ctan.org/pkg/enumerate
\usepackage{enumitem}

\usepackage{appendix}

\newtheorem{thm}{Theorem}[section]
\newtheorem{prop}[thm]{Proposition}
\newtheorem{theorem}[thm]{Theorem}

\newtheorem{remark}[thm]{\it Remark}
\newtheorem{example}[thm]{\it Example}

\usepackage{tabularray}

\usepackage{tikz-3dplot}
\usepackage{xifthen}
\usepackage{caption}

\tdplotsetmaincoords{60}{125}
\tdplotsetrotatedcoords{0}{20}{0} %<- rotate around (z,y,z)

\usepackage[
labelfont=sf,
hypcap=false,
format=hang
]{caption}

\usepackage{mathtools}

\begin{document}

%\title{Matrix re-factorization problems: Pentagon  and reverse-Pentagon maps}
\title{Entwining tetrahedron  maps}

\author[P. Kassotakis]{Pavlos Kassotakis}
\address{Pavlos Kassotakis, Department of Mathematical Methods in Physics, Faculty of Physics,
University of Warsaw, Pasteura 5, 02-093, Warsaw, Poland}
 \email{Pavlos.Kassotakis@fuw.edu.pl, pavlos1978@gmail.com}

\begin{abstract} 
We present three non-equivalent procedures to obtain {\em entwining (non-constant)  tetrahedron maps}. Given a tetrahedron map, the first procedure incorporates its underlying symmetry group. With the second procedure we obtain several classes of entwining tetrahedron maps by  considering certain compositions of {\em pentagon} with {\em reverse-pentagon maps} which  satisfy certain compatibility relations the so-called {\em ten-term relations}. Using the third procedure, provided that a given  tetrahedron map  admits at least one {\em companion map (partial inverse)}, we
obtain  entwining  set theoretical solutions of the tetrahedron equation.
\end{abstract}

\maketitle

% Uncomment for keywords
\vspace{2pc}

%\setcounter{tocdepth}{1}

%\tableofcontents

%%%%%%%%%%%%%%%%%%%%%%%%
%%%%%%%%%%%%%%%%%%%%%%%%
%%%%%%%%%%%%%%%%%%%%%%%%

\section{Introduction}
The {\em tetrahedron equation} was introduced by Zamolodchikov  as the factorization condition for the relativistic S-matrix of the scattering of straight strings in $(2+ 1)-$dimensional space-time and in association with a three-dimensional exactly solvable lattice-model \cite{Zamolodjicov_1980,Zamolodchikov:1981}. It  can be realized as a generalization of the Yang-Baxter equation  and explicitly  reads
\begin{align} \label{tetra}
T_{123}T_{145}T_{246}T_{356}=T_{356}T_{246}T_{145}T_{123},
\end{align}
where the  objects ($T_{ijk}$) represent operators or maps; so we  respectively have the ``operator" or the ``set-theoretic" version of (\ref{tetra}). In the set-theoretic version of (\ref{tetra}) the subscripts denote the sets where the map
\begin{align*}
T: \underbrace{\mathbb{X}\times \mathbb{X}\times \mathbb{X}}_{\mathbb{X}^3}\rightarrow \underbrace{\mathbb{X}\times \mathbb{X}\times \mathbb{X}}_{\mathbb{X}^3}
\end{align*}
 acts non-trivially when is acting on $ \mathbb{X}^6$. Solutions of the set-theoretic version of (\ref{tetra}) are called {\em set-theoretical solutions of the tetrahedron equation} or simply {\em tetrahedron maps}. 
 
The first classification results on tetrahedron maps defined on $\mathbb{CP}^1\times \mathbb{CP}^1\times \mathbb{CP}^1,$ were  based on the local Yang-Baxter equation \cite{Maillet:1989a,Maillet:1989} and  appeared in  \cite{Sergeev-1998,Korepanov-1998}. Moreover, the connection of tetrahedron maps of certain type with three-dimensional integrable partial difference equations (P$\Delta$E's) \cite{Nijhoff:1990} was established in \cite{Kashaev:1996} and complemented in \cite{Kass1}; see also \cite{Kassotakis_2018} where compatible three dimensional systems of P$\Delta$E's defined on higher order stencils were associated with tetrahedron maps. For recent developments on tetrahedron  maps see \cite{Doliwa:2020,Igonin:2022,Rizos:2022,Inoue_2024}, while for  further connections of the tetrahedron equation to various areas of mathematics and physics we refer to \cite{Stro:1997,Dimakis:2015,Talalaev_2021}.

In this article we consider the {\em entwining or non-constant} form of (\ref{tetra}), that is 
\begin{align} \label{etetra}
T^{(1)}_{123} T^{(2)}_{145} T^{(3)}_{246} T^{(4)}_{356}=T^{(4)}_{356} T^{(3)}_{246} T^{(2)}_{145} T^{(1)}_{123}.
\end{align}
The superscripts denote that the operators or maps that participate on the equation above might differ. Although the tetrahedron equation (\ref{tetra}) is well studied, its entwining counterpart  (\ref{etetra}) is not. We can  recall though the studies in \cite{Korepanov:2024}, where entwining solutions of a specific kind of the operator form of (\ref{etetra}) are considered. For solutions of the set-theoretic version of (\ref{etetra}), we can only recall  isolated examples, f.i. see \cite{Kassotakis:2019,Kass1}.

In this article we present three non-equivalent procedures to obtain entwining tetrahedron maps. The  first procedure is introduced in Section \ref{sec2.1} where we show how to construct entwining tetrahedron maps from a given tetrahedron map that possesses a group of discrete symmetries. In detail, this procedure serves as an  enrichment of  the procedure of obtaining entwining tetrahedron maps from symmetries that was firstly introduced in \cite{Kassotakis:2019,Kass1}. In Section \ref{sec3},  we show how to obtain entwining tetrahedron maps from a pentagon and a reverse pentagon map that satisfy a consistency relation the so-called {\em ten-term relation} \cite{Kashaev:1998} or a variant of the  ten-term relation that we refer to as {\em opposite ten-term-relation}. We are also extending the aforementioned construction applied to entwining pentagon and reverse-pentagon maps. Finally in Section \ref{sec4},  provided that a tetrahedron map admits at least one companion map (partial inverse) we construct entwining tetrahedron maps. This procedure serves as a direct extension of some of the results of \cite{CKT} to the tetrahedron case.

\section{Tetrahedron maps their symmetries and entwining tetrahderon maps}\label{sec2.1}
 Let $\mathbb{X}$ be a set. We will need the following definitions.
\begin{definition} \label{def01}
Two  maps $S: \mathbb{X} \times \mathbb{X}\times \mathbb{X} \rightarrow \mathbb{X} \times \mathbb{X}\times \mathbb{X}$ and
$T: \mathbb{X} \times \mathbb{X}\times \mathbb{X} \rightarrow \mathbb{X} \times \mathbb{X}\times \mathbb{X}$ will be called
$M\ddot{o}b$ equivalent if there exists a bijection $\phi: \mathbb{X} \rightarrow \mathbb{X}$ such that
$ T\circ  (\phi\times \phi\times \phi)= (\phi\times \phi \times \phi)\circ S.$
\end{definition}

\begin{definition}
A map $T: \mathbb{X} \times \mathbb{X}\times \mathbb{X} \rightarrow \mathbb{X} \times \mathbb{X}\times \mathbb{X},$ will be called {\em tetrahedron map} if it satisfies the {\em tetrahedron equation}
\begin{align*}
T_{123}\circ T_{145}\circ T_{246}\circ T_{356}=T_{356}\circ T_{246}\circ T_{145}\circ T_{123},
\end{align*}
where $T_{ijk},$ $i\neq j\neq k\in \{1,\ldots, 6\},$ denote the maps that act as $T$ on the $i^{th},$ $j^{th},$ and $k^{th},$ components of $\mathbb{X}^6$ and as identity to the remaining ones.
\end{definition}

The importance of the equivalence relation of Definition \ref{def01} is given in the following  proposition.

\begin{prop}
Let $T: \mathbb{X} \times \mathbb{X}\times \mathbb{X} \rightarrow \mathbb{X} \times \mathbb{X}\times \mathbb{X}$ be a tetrahedron map and $S$
a $M\ddot{o}b$ equivalent map to $T$. Then $S$ is also a tetrahedron  map.
\end{prop}
\begin{proof}
This is a well known Proposition. For a proof see  \cite{Kass1}.
\end{proof}

\begin{definition} \label{def_sym}
A map $\phi: \mathbb{X} \rightarrow \mathbb{X}$ will be called  {\em symmetry} of the  map $T: \mathbb{X} \times \mathbb{X}\times \mathbb{X} \rightarrow \mathbb{X} \times \mathbb{X}\times \mathbb{X}$ if it holds
$T\circ  (\phi\times \phi\times \phi)= (\phi\times \phi \times \phi)\circ T.$
\end{definition}

\begin{remark} \label{rem0}
Let $\tau:\mathbb{X} \times \mathbb{X} \rightarrow \mathbb{X} \times \mathbb{X},$ be the transposition map, that is $\tau:(x,y)\mapsto (y,x).$ Then it is easy to verify the following:
\begin{enumerate}
  \item A map  $T: \mathbb{X} \times \mathbb{X}\times \mathbb{X} \rightarrow \mathbb{X} \times \mathbb{X}\times \mathbb{X}$ is a tetrahedron map iff 
      \begin{align*}
      T^r:=(id\times \tau)\circ (\tau\times id)\circ(id\times \tau)\circ T\circ(id\times \tau)\circ (\tau\times id)\circ(id\times \tau),
       \end{align*}
       is a tetrahedron map. Furthermore, mapping $T$ will be called reversible if it holds 
       \begin{align*}
       T\circ T^r=id.
       \end{align*}
       For the mappings $T_{ijk},$ $i<j<k,$ there is $T^r_{ijk}=\tau_{ik}\circ T_{ijk}\circ \tau_{ik}=T_{kji},$ where as usual $\tau_{ik}$ denotes that map that acts as  $\tau$ in the $i^{th}$ and in the $k^{th}$ factor of $\mathbb{X} \times \mathbb{X}\times \ldots\times \mathbb{X}.$ 
  \item An invertible map  $T: \mathbb{X} \times \mathbb{X}\times \mathbb{X} \rightarrow \mathbb{X} \times \mathbb{X}\times \mathbb{X}$ is a tetrahedron map iff $T^{-1}$ is a tetrahedron map.
  \item A map  $T: \mathbb{X} \times \mathbb{X}\times \mathbb{X} \rightarrow \mathbb{X} \times \mathbb{X}\times \mathbb{X}$ is a tetrahedron map iff $(\phi\times id\times \phi)\circ  T\circ(id\times \phi^{-1}\times id)$ is a tetrahedron map, where  $\phi: \mathbb{X} \rightarrow \mathbb{X}$ is a symmetry of  $T.$ 
\end{enumerate} 
\end{remark}

\begin{definition}
The maps $T^{(a)}: \mathbb{X} \times \mathbb{X}\times \mathbb{X} \rightarrow \mathbb{X} \times \mathbb{X}\times \mathbb{X},$ $a\in\{1,\ldots, 4\},$ will be called {\em entwining tetrahedron map} or {\em non-constant tetrahedron maps} if they satisfy the  {\em entwining (or non-constant)  tetrahedron equation}
\begin{align*}
T_{123}^{(1)}\circ T_{145}^{(2)}\circ T_{246}^{(3)}\circ T_{356}^{(4)}=T_{356}^{(4)}\circ T_{246}^{(3)}\circ T_{145}^{(2)}\circ T_{123}^{(1)},
\end{align*}
where $T_{ijk}^{(a)},$ $i\neq j\neq k\in \{1,\ldots, 6\},$ denote the maps that act as $T^{(a)}$ on the $i^{th},$ $j^{th},$ and $k^{th},$ components of $\mathbb{X}^6$ and as identity to the remaining ones.
\end{definition}

In \cite{Kassotakis:2019,Kass1} it was show how to obtain entwining tetrahedron maps from a given tetrahedron map that admits a symmetry. In what follows we extend the results of the aforementioned articles. We will need the following definition.
\begin{definition}
Let $T: \mathbb{X} \times \mathbb{X}\times \mathbb{X} \rightarrow \mathbb{X} \times \mathbb{X}\times \mathbb{X}$ be a tetrahedron map that admits a symmetry $\phi: \mathbb{X} \rightarrow \mathbb{X}.$  Composite mappings of the form $T_{ijk}\circ \phi_l,$ 
or $\phi_l\circ T_{ijk},$ $l\in\{i,j,k\}$ which participate in an entwining tetrahedron equation, will be called {\em entwining tetrahedron maps of  first-type} and they will be denoted as
\begin{align*}
T_{\overline{i}jk}:=&T_{ijk}\circ \phi_i,&T_{i\overline{j}k}:=&T_{ijk}\circ \phi_j,&T_{ij\overline{k}}:=&T_{ijk}\circ \phi_k,\\
T_{\underline{i}jk}:=&\phi_i\circ T_{ijk},&T_{i\underline{j}k}:=&\phi_j\circ T_{ijk},&T_{ij\underline{k}}:=&\phi_k\circ T_{ijk}.
\end{align*} 
Composite mappings of the form $T_{ijk}\circ \phi_l\circ \phi_m,$ 
or $\phi_l\circ \phi_m\circ T_{ijk},$  or $\phi_l\circ T_{ijk}\circ \phi_m,$ $l\neq m\in\{i,j,k\},$ will be called {\em entwining tetrahedron maps of  second-type} and they will be denoted as
\begin{align*}
T_{\overline{i}j\overline{k}}:=&T_{ijk}\circ \phi_i\circ \phi_k,&T_{\underline{i}\underline{j}k}:=&\phi_i\circ \phi_j\circ T_{ijk}, &T_{\underline{i}\overline{j}k}:=&\phi_i\circ T_{ijk}\circ \phi_j,& \mbox{etc.}
\end{align*}
Similarly are defined {\em entwining tetrahedron maps of  third-type}.
\end{definition}
The next Theorem incorporates entwining tetrahedron maps of first-type. We leave the studies on entwining tetrahedron maps of the second and the third-type for  future work.
\begin{theorem}\label{theo0}
Let $T: \mathbb{X} \times \mathbb{X}\times \mathbb{X} \rightarrow \mathbb{X} \times \mathbb{X}\times \mathbb{X}$ be a tetrahedron map that admits a %\footnote{We assume involutivity in order to avoid degenerate cases f.i. entwining equations that include expressions of the form $T_{\underline{i}jk}\circ T_{\overline{i}mn}$}
 symmetry $\phi: \mathbb{X} \rightarrow \mathbb{X}.$  Then the entwining tetrahedron maps of first-type
$T_{\overline{i}jk},$ $T_{i\overline{j}k},$ $T_{ij\overline{k}},$ $T_{\underline{i}jk},$ $T_{i\underline{j}k},$ and $T_{ij\underline{k}},$
satisfy the following entwining tetrahedron equations
\begin{align}\label{es1}
  T_{123}\circ T_{\overline{1}45}\circ T_{\overline{2}46}\circ T_{\overline{3}56}=& T_{\overline{3}56}\circ T_{\overline{2}46}\circ T_{\overline{1}45}\circ T_{123}, \\ \label{es2}
  T_{123}\circ T_{\underline{1}45}\circ T_{\underline{2}46}\circ T_{\underline{3}56}=& T_{\underline{3}56}\circ T_{\underline{2}46}\circ T_{\underline{1}45}\circ T_{123},\\ \label{es3}
  T_{\overline{1}23}\circ T_{145}\circ T_{2\underline{4}6}\circ T_{3\underline{5}6}=& T_{3\underline{5}6}\circ T_{2\underline{4}6}\circ T_{145}\circ T_{\overline{1}23}, \\ \label{es4}
T_{1\overline{2}3}\circ T_{1\overline{4}5}\circ T_{246}\circ T_{35\underline{6}}=& T_{35\underline{6}}\circ T_{246}\circ T_{1\overline{4}5}\circ T_{1\overline{2}3}, \\ \label{es5}
T_{12\overline{3}}\circ T_{14\overline{5}}\circ T_{24\overline{6}}\circ T_{356}=& T_{356}\circ T_{24\overline{6}}\circ T_{14\overline{5}}\circ T_{12\overline{3}}, \\ \label{es6}
 T_{\underline{1}23}\circ T_{145}\circ T_{2\overline{4}6}\circ T_{3\overline{5}6}=& T_{3\overline{5}6}\circ T_{2\overline{4}6}\circ T_{145}\circ T_{\underline{1}23}, \\ \label{es7}
T_{1\underline{2}3}\circ T_{1\underline{4}5}\circ T_{246}\circ T_{35\overline{6}}=& T_{35\overline{6}}\circ T_{246}\circ T_{1\underline{4}5}\circ T_{1\underline{2}3}, \\ \label{es8}
T_{12\underline{3}}\circ T_{14\underline{5}}\circ T_{24\underline{6}}\circ T_{356}=& T_{356}\circ T_{24\underline{6}}\circ T_{15\underline{5}}\circ T_{12\underline{3}}.
\end{align}
\end{theorem}
\begin{proof}
By using the definition of the symmetry (see Definition \ref{def_sym}) we move all $\phi_i$ to the left or to the right on each equation (\ref{es1})-(\ref{es8}). Then the equation under consideration is verified due to the assumption that $T$ is a tetrahedron map.

Indeed, from (\ref{es1}) we respectively have
\begin{align*}
T_{123}\circ T_{\overline{1}45}\circ T_{\overline{2}46}\circ T_{\overline{3}56}-&T_{\overline{3}56}\circ T_{\overline{2}46}\circ T_{\overline{1}45}\circ T_{123}\\
= T_{123}\circ T_{145}\circ \phi_1\circ T_{246}\circ \phi_2\circ T_{356}\circ \phi_3-&T_{356}\circ \phi_3\circ T_{246}\circ \phi_2\circ T_{145}\circ \phi_1\circ T_{123}\\
=T_{123}\circ T_{145}\circ T_{246}\circ \circ T_{356}\circ \phi_1\circ \phi_2\circ \phi_3-&T_{356}\circ T_{246}\circ  T_{145}\circ \underbrace{\phi_1\circ \phi_2\circ \phi_3\circ T_{123}}_{=T_{123}\circ \phi_1\circ \phi_2\circ \phi_3}\\
=\left(T_{123}\circ T_{145}\circ T_{246}\circ \circ T_{356}-\right.&\left.T_{356}\circ T_{246}\circ  T_{145}\circ\circ T_{123}\right)\circ\phi_1\circ \phi_2\circ \phi_3=0,
\end{align*}
where we have used that $\phi$ is a symmetry and the assumption that  $T$ is a tetrahedron map. In a similar manner follows the proof of the remaining entwining equations (\ref{es2})-(\ref{es8}).
\end{proof}
\begin{remark}
From item $(2)$ of  Remark (\ref{rem0}), we deduce that the inverses of the entwining tetrahedron maps  $T_{\overline{i}jk},T_{i\overline{j}k}, T_{i\underline{j}k}, etc.$ are also entwining tetrahedron maps that satisfy the inverse  equations of (\ref{es1})-(\ref{es8}). 
\end{remark}

\begin{example}
The Hirota map was introduced in \cite{Zelevinsky_1996} c.f. \cite{Korepanov-1998,Sergeev-1998}, it is a tetrahedron map that is related to the Hirota-Miwa equation \cite{Kashaev:1996}. In detail it  reads
\begin{align} \label{hirota0}
T_{H}:\mathbb{CP}^1\times\mathbb{CP}^1\times\mathbb{CP}^1\ni(x,y,z)\rightarrow (u,v,w)\in \mathbb{CP}^1\times\mathbb{CP}^1\times\mathbb{CP}^1,
\end{align}
where
\begin{align*}
  u=&\frac{xy}{x+z},&v=&x+z,&w=&\frac{yz}{x+z},
\end{align*}
and $\mathbb{CP}^1=\mathbb{C}\cup \{\infty\},$ the Reimann sphere.
It was shown in \cite{Kassotakis:2019} that it admits the symmetry $\phi: \mathbb{CP}^1\ni x\mapsto -x \in \mathbb{CP}^1,$ so according to Theorem \ref{theo0}, the enwining tetrahedron maps $T_{\overline{i}jk},$ $T_{i\overline{j}k},$ $T_{ij\overline{k}},$ $T_{\underline{i}jk},$ $T_{i\underline{j}k},$ and $T_{ij\underline{k}},$
satisfy the entwining tetrahedron equations (\ref{es1})-(\ref{es8}). For example, mapping $T:=T_{H},$ entwines with $S:=T\circ (\phi\times id\times id),$ and $R:=(id\times \phi \times id)\circ T,$ to satisfy (\ref{es3}) that reads
\begin{align*}
  S_{123}\circ T_{145}\circ R_{246}\circ R_{356}=&R_{356}\circ R_{246}\circ T_{145}\circ S_{123}.
\end{align*}
 The maps $S$ and $R$ respectively read
\begin{align*}
S:(x,y,z)\mapsto& \left(\frac{xy}{x-z},z-x,\frac{yz}{z-x}\right),& R:(x,y,z)\mapsto & \left(\frac{xy}{x+z},-x-z,\frac{yz}{x+z}\right).
\end{align*}. 
\end{example}
\section{Ten-term relations and entwining tetrahedron maps} \label{sec3}

In this Section first we recall \cite{Kashaev:1998} where it is shown how to construct tetrahedron maps from a pentagon together with  a reverse pentagon map that satisfy a compatibility relation the so-called {\em ten-term relation}. Furthermore, we show how to construct tetrahedron maps by using a variant of the ten-term relation that we refer to as {\em opposite ten-term relation}.  We also show how to construct entwining tetrahedron maps from  pentagon and reverse pentagon maps that satisfy the ten-term relation or its variant.  

%\subsection{{\color{blue} From  pentagon maps to entwining tetrahedron maps}}

First we will need the following definitions.
\begin{definition}[Pentagon and reverse-pentagon maps \cite{Biedenharn:1953,Elliott:1953,Drinfeld_p,MoorSeib:89,Maillet:1990,Kaufmann_1993}]
Let $\mathbb{X}$ be a set. The map $S: \mathbb{X} \times \mathbb{X}\rightarrow  \mathbb{X} \times \mathbb{X},$
 will be called {\em pentagon map} %\cite{Zakrzewski:1992,Skandalis:1993}
 if it satisfies the {\em  pentagon equation}  
\begin{align*}
S_{12}\circ S_{13}\circ S_{23}= S_{23}\circ S_{12},
\end{align*}
while the map $\bar S: \mathbb{X} \times \mathbb{X}\rightarrow  \mathbb{X} \times \mathbb{X},$
 will be called {\em reverse-pentagon map} if it satisfies the {\em  reverse-pentagon equation}
\begin{align*}
\bar S_{23}\circ \bar S_{13}\circ \bar S_{12}= \bar S_{12}\circ \bar S_{23}.
\end{align*}
\end{definition}

\begin{definition}[Ten-term relation \cite{Kashaev:1998}] \label{ten-term}
The maps $W$ and $\bar W$ that map $\mathbb{X} \times \mathbb{X}$ to itself, will be said to satisfy the {\em  ten-term relation} if it holds
\begin{align} \label{ten-term}
\bar W_{13}\circ W_{12}\circ \bar W_{14}\circ W_{34}\circ \bar W_{24}=
W_{34}\circ \bar W_{24}\circ W_{14}\circ \bar W_{13}\circ W_{12}.
\end{align}
\end{definition}

\begin{theorem}[\cite{Kashaev:1998}] \label{Def:Tetra-Penta}
Let $S$ respectively $\bar S$ be a pentagon respectively a reverse pentagon map that satisfy the {\em ten-term} relation (\ref{ten-term}) and they are both invertible. Then the map
\begin{align} \label{tet_ks}
T:=\bar S_{13}\circ\tau_{23}\circ S_{13},
\end{align}
is a  tetrahedron map. 
\end{theorem}
In the theorem above  $\tau:\mathbb{X} \times \mathbb{X}\rightarrow  \mathbb{X} \times \mathbb{X}$ denotes the  transposition map. That is, $\tau_{ij}$ stands for the transposition of the $i-$th and $j-$th arguments of $\mathbb{X}^N,$ $N\in\mathbb{N}$. For example $\tau_{13}: (x,y,z)\mapsto (z,y,x).$

We are now ready to present the main Theorem of this section, where   we extend Theorem  \ref{Def:Tetra-Penta}, so that to incorporate  entwining tetrahedron equations. Also we relax the assumption of invertibility of the maps $S$ and $\bar S$.  We will need the following definition.
\begin{definition}[Opposite ten-term relation] \label{ten-term22}
The maps $W$ and $\bar W$ that map $\mathbb{X} \times \mathbb{X}$ to itself, will be said to satisfy the {\em opposite ten-term relation} if it holds
\begin{align} \label{ten-term2}
 W_{24}\circ \bar W_{34}\circ  W_{14}\circ \bar W_{12}\circ  W_{13}=
\bar W_{12}\circ  W_{13}\circ \bar W_{14}\circ  W_{24}\circ \bar W_{34}.
\end{align}
\end{definition}
\begin{remark}
When there is 
\begin{align}\label{equiv}
W\circ \bar W=\bar W\circ W=id,
\end{align}
then the opposite ten-term relation (\ref{ten-term2}), is equivalent to  the   ten-term relation (\ref{ten-term}).   Note that from (\ref{equiv}) we deduce that either $W$ is invertible with $\bar W:=W^{-1},$  or $W$ is reversible ($W_{ij}\circ W_{ji}=id$) with $\bar W:=\tau\circ W\circ \tau.$
\end{remark}

\begin{theorem} \label{theo2}
   Let $S$ respectively $\bar S$ be a pentagon respectively a reverse-pentagon map that satisfy the  ten-term relation (\ref{ten-term}), 
   or the  opposite ten-term relation (\ref{ten-term2}).
   Then the maps
   \begin{align*}
   T^{(1)}:=&\bar S_{13}\circ\tau_{23}\circ S_{13},& T^{(2)}:=&S_{13}\circ\tau_{12}\circ \bar S_{13},
   \end{align*}
   satisfy the  tetrahedron equations
   \begin{gather}\label{en01}
T^{(1)}_{123}\circ T^{(1)}_{145}\circ T^{(1)}_{246}\circ T^{(1)}_{356}=T^{(1)}_{356}\circ T^{(1)}_{246}\circ T^{(1)}_{145}\circ T^{(1)}_{123},\\ \label{en02}
T^{(2)}_{123}\circ T^{(2)}_{145}\circ T^{(2)}_{246}\circ T^{(2)}_{356}=T^{(2)}_{356}\circ T^{(2)}_{246}\circ T^{(2)}_{145}\circ T^{(2)}_{123}.
\end{gather}
While together with the maps 
\begin{align*}
    T^{(3)}:=&S_{12}\circ\tau_{23}\circ \bar S_{12},& T^{(4)}:=&\bar S_{23}\circ\tau_{12}\circ S_{23},
     \end{align*}
   satisfy the following entwining tetrahedron equations
   \begin{gather} \label{enn03}
   T^{(1)}_{123}\circ T^{(3)}_{145}\circ T^{(3)}_{246}\circ T^{(3)}_{356}=T^{(3)}_{356}\circ T^{(3)}_{246}\circ T^{(3)}_{145}\circ T^{(1)}_{123},\\ \label{enn04} 
   T^{(4)}_{123}\circ T^{(4)}_{145}\circ T^{(4)}_{246}\circ T^{(2)}_{356}=T^{(2)}_{356}\circ T^{(4)}_{246}\circ T^{(4)}_{145}\circ T^{(4)}_{123},\\ \label{enn05} 
T^{(3)}_{123}\circ T^{(3)}_{145}\circ T^{(4)}_{246}\circ T^{(4)}_{356}=T^{(4)}_{356}\circ T^{(4)}_{246}\circ T^{(3)}_{145}\circ T^{(3)}_{123}.
\end{gather}
\end{theorem}

\begin{proof}
The proof follows by algebraic manipulations. First we substitute the appropriate  $T^{(i)}$, $i=1,\ldots,4$ to each of the tetrahedron equations  (\ref{en01})-(\ref{enn05}), then by using the identities 
\begin{align} \label{sube}
\tau_{ij}\circ S_{ik}=&S_{jk}\circ \tau_{ij}, &\tau_{ij}\circ S_{ij}\circ\tau_{ij}=&S_{ji}, &i\neq j\neq k,
\end{align}
we move the transpositions $\tau_{ij}$ to the right of the lhs and of the rhs of the tetrahedron equation under consideration. Finally, by using the fact that $S_{ij}$ is a pentagon map and $\bar S_{ij}$ is a reverse pentagon map, we arrive to  the ten-term relation  (\ref{ten-term}) or to the variant ten-term relatiom (\ref{ten-term2}).

In the respect above we now prove that $T^{(1)}$ is a tetrahedron map. In terms of $T^{(1)}$ the tetrahedron equation (\ref{en01}) reads
\begin{multline*}
 \bar S_{13}\circ \tau_{23}\circ  S_{13}\circ \bar S_{15}\circ \tau_{45}\circ S_{15}\circ \bar S_{26}\circ \tau_{46}\circ  S_{26}\circ \bar S_{36}\circ \tau_{56}\circ  S_{36}\\
=\bar S_{36}\circ \tau_{56}\circ  S_{36}\circ \bar S_{26}\circ \tau_{46}\circ  S_{26}\circ \bar S_{15}\circ \tau_{45}\circ S_{15}\circ \bar S_{13}\circ \tau_{23}\circ  S_{13}.
\end{multline*}
We use (\ref{sube}) to move the transpositions $\tau_{ij}$ to the right to obtain
\begin{multline*}
\underbrace{ \bar S_{13}\circ \bar S_{36}}_{=\bar S_{36}\circ \bar S_{16}\circ \bar S_{13}}\circ  S_{12}\circ \bar S_{15}\circ S_{35}\circ \bar S_{25}\circ S_{14}\circ  S_{24}\circ \tau_{23}\circ \tau_{45}\circ \tau_{46}\circ \tau_{56}\\
=\bar S_{36}\circ \bar S_{16}\circ  S_{35}\circ \bar S_{25}\circ S_{15}\circ\bar  S_{13}\circ \underbrace{ S_{24}\circ  S_{12}}_{=S_{12}\circ  S_{14}\circ S_{24}}\circ \tau_{56}\circ \tau_{46}\circ \tau_{45}\circ \tau_{23}.
\end{multline*}
 Using the fact that it holds $\tau_{45}\circ \tau_{46}\circ \tau_{56}=\tau_{56}\circ \tau_{46}\circ \tau_{45},$  and the fact that  $S$ is a pentagon  and $\bar S$ a reverse pentagon map,    the equation above reads
 \begin{align*}
 \bar S_{36}\circ \bar S_{16}\circ\left(\bar S_{13}\circ S_{12}\circ\bar S_{15}\circ S_{35}\circ\bar S_{25}-S_{35}\circ \bar S_{25}\circ S_{15}\circ\bar S_{13}\circ  S_{12}\right)\circ S_{14}\circ  S_{24}\circ K=0,
 \end{align*}
where $K:=\tau_{23}\circ \tau_{45}\circ \tau_{46}\circ \tau_{56},$ that holds true since the expression inside the parenthesis is the  ten-term relation (\ref{ten-term}) (expressed in $1,2,3,5$ indices) that is satisfied according to the assumptions of the theorem. So indeed $T^{(1)}$ is a tetrahedron map.

We now prove that $T^{(2)}$ is a tetrahedron map. In terms of $T^{(2)}$ the tetrahedron equation (\ref{en02}) reads
\begin{multline*}
 S_{13}\circ \tau_{12}\circ \bar S_{13}\circ S_{15}\circ \tau_{14}\circ\bar S_{15}\circ S_{26}\circ \tau_{24}\circ \bar S_{26}\circ S_{36}\circ \tau_{35}\circ \bar S_{36}\\
=S_{36}\circ \tau_{35}\circ \bar S_{36}\circ S_{26}\circ \tau_{24}\circ \bar S_{26}\circ S_{15}\circ \tau_{14}\circ\bar S_{15}\circ  S_{13}\circ \tau_{12}\circ \bar S_{13}.
\end{multline*}
We use (\ref{sube}) to move the transpositions $\tau_{ij}$ to the right to obtain
\begin{multline*}
 S_{13}\circ  S_{16}\circ \bar S_{23}\circ S_{25}\circ\bar S_{26}\circ S_{36}\circ \underbrace{\bar S_{45}\circ \bar S_{56}}_{=\bar S_{56}\circ \bar S_{46}\circ \bar S_{45}}\circ \tau_{12}\circ \tau_{14}\circ \tau_{24}\circ \tau_{35}\\
=\underbrace{S_{36}\circ  S_{13}}_{= S_{13}\circ  S_{16}\circ S_{36}}\circ \bar S_{56}\circ S_{26}\circ\bar S_{23}\circ  S_{25}\circ \bar S_{46}\circ \bar S_{45}\circ \tau_{35}\circ \tau_{24}\circ \tau_{14}\circ \tau_{12}.
\end{multline*}
 Using the fact that it holds $\tau_{12}\circ \tau_{14}\circ \tau_{24}=\tau_{24}\circ \tau_{14}\circ \tau_{12},$  and the fact that  $S$ is a pentagon  and $\bar S$ a reverse pentagon map,    the equation above reads
 \begin{align*}
 S_{13}\circ  S_{16}\circ\left(\bar S_{23}\circ S_{25}\circ\bar S_{26}\circ S_{36}\circ\bar S_{56}-S_{36}\circ \bar S_{56}\circ S_{26}\circ\bar S_{23}\circ  S_{25}\right)\circ\bar S_{46}\circ \bar S_{45}\circ M=0,
 \end{align*}
where $M:=\tau_{12}\circ\tau_{14}\circ \tau_{24}\circ \tau_{35},$ that holds true since the expression inside the parenthesis is the variant ten-term relation (\ref{ten-term}) (expressed in $2,3,5,6$ indices) that is satisfied according to the assumptions of the theorem. So indeed $T^{(2)}$ is a tetrahedron map.
 
  The proof of the rest of the theorem  follows in a similar manner. 
 
  For example, let us use $T^{(1)},T^{(3)}$ to prove  (\ref{enn03}). Working similarly as above, in terms of $T^{(1)},$ and $T^{(3)},$ from the entwining tetrahedron equation (\ref{en04}) we have
%\begin{multline*}
%\bar S_{13}\circ \tau_{23}\circ S_{13}\circ S_{14}\circ \tau_{45}\circ\bar S_{14}\circ S_{24}\circ \tau_{46}\circ \bar S_{24}\circ S_{35}\circ %\tau_{56}\circ \bar S_{35}\\
%=S_{35}\circ \tau_{56}\circ \bar S_{35}\circ S_{24}\circ \tau_{46}\circ \bar S_{24}\circ S_{14}\circ \tau_{45}\circ\bar S_{14}\circ\bar S_{13}\circ %\tau_{23}.
%\end{multline*}
%We use (\ref{sube}) to move the transpositions $\tau_{ij}$ to the right to obtain
\begin{multline*}
\bar S_{13}\circ  \underbrace{S_{12}\circ S_{14}\circ S_{24}}_{=S_{24}\circ S_{12}}\circ\bar S_{15}\circ S_{35}\circ \bar S_{36}\circ \bar S_{25}\circ \tau_{23}\circ \tau_{45}\circ \tau_{46}\circ \tau_{56}\\
=S_{24}\circ  S_{35}\circ \bar S_{25}\circ S_{15}\circ\underbrace{\bar S_{36}\circ \bar S_{16}\circ \bar S_{13}}_{=\bar S_{13}\circ \bar S_{36}}\circ S_{12}\circ \tau_{56}\circ \tau_{46}\circ \tau_{45}\circ \tau_{23},
\end{multline*}
%where we used that $S$ is a pentagon  and $\bar S$ a reverse pentagon map. Using the fact that it holds $\tau_{45}\circ \tau_{46}\circ %\tau_{56}=\tau_{56}\circ \tau_{46}\circ \tau_{45},$ 
The equation above after some cancelations reads
\begin{align*}
\bar S_{13}\circ S_{12}\circ\bar S_{15}\circ  S_{35}\circ \bar S_{25}=S_{35}\circ \bar S_{25}\circ S_{15}\circ \bar S_{13}\circ  S_{12},
\end{align*}
that is equivalent to the ten term relation (\ref{ten-term}).

%Working similarly, we arrive to the conclusion that $T^{(2)}$ is a tetrahedron map provided that  the variant ten-term relation (\ref{ten-term2}) is %satisfied.
%Furthermore, following exactly the same procedure for the remaining equations of (\ref{en01})-(\ref{en08}), we obtain (\ref{ten-term}) or (\ref{ten-term2}) %and that completes the proof.
\end{proof}

\begin{remark}
Note that  the proof that $T^{(1)},$ $T^{(2)}$ are tetrahedron maps in the theorem above , does not requires the invertibility of the pentagon map $S$ and the reverse pentagon map $\bar S$.
\end{remark}
%\begin{remark}
%When the reverse pentagon map $\bar S$ that participates in the Theorem (\ref{theo2}) is the inverse of the pentagon map $S,$ then the variant ten-term %relation (\ref{ten-term2}), is equivalent to  the   ten-term relation (\ref{ten-term}). Actually then (\ref{ten-term2}) coincides with  the inverse of %(\ref{ten-term}). In all other cases, the ten term relations (\ref{ten-term}) and (\ref{ten-term2}) are not equivalent.  
%\end{remark}
When the pentagon map  and the reverse pentagon map satisfy the conditions (\ref{equiv}), additional entwining tetrahedron maps can be generated as the following proposition suggests.
\begin{prop} \label{propsec3}
Consider the maps $T^{(i)},$ $i=1,\ldots, 4$ of Theorem \ref{theo2}. Then the maps
\begin{align*}
\hat T^{(1)}:=&\tau_{23}\circ T^{(3)}\circ \tau_{23}=S_{13}\circ\tau_{23}\circ \bar S_{13},& \hat T^{(2)}:=&\tau_{12}\circ T^{(4)}\circ \tau_{12}=\bar S_{13}\circ\tau_{12}\circ  S_{13},\\
\hat T^{(3)}:=&\tau_{23}\circ T^{(1)}\circ \tau_{23}=\bar S_{12}\circ\tau_{23}\circ S_{12}& \hat T^{(4)}:=&\tau_{12}\circ T^{(1)}\circ \tau_{12}=\bar S_{23}\circ\tau_{13}\circ  S_{23}\\
\hat T^{(5)}:=&\tau_{23}\circ T^{(2)}\circ \tau_{23}=S_{12}\circ\tau_{13}\circ \bar S_{12}, & \hat T^{(6)}:=&\tau_{12}\circ T^{(2)}\circ \tau_{12}=S_{23}\circ\tau_{12}\circ \bar S_{23},
\end{align*} 
 provided that (\ref{equiv}) holds, satisfy the following entwining tetrahedron equations
   \begin{gather} \label{en03}
\hat T^{(1)}_{123}\circ \hat T^{(1)}_{145}\circ \hat T^{(2)}_{246}\circ \hat T^{(2)}_{356}=\hat T^{(2)}_{356}\circ\hat T^{(2)}_{246}\circ\hat T^{(1)}_{145}\circ\hat T^{(1)}_{123},\\ \label{en04}
\hat T^{(3)}_{123}\circ \hat T^{(3)}_{145}\circ \hat T^{(4)}_{246}\circ \hat T^{(4)}_{356}=\hat T^{(4)}_{356}\circ\hat T^{(4)}_{246}\circ\hat T^{(3)}_{145}\circ\hat T^{(3)}_{123},\\ \label{en05}
\hat T^{(5)}_{123}\circ \hat T^{(5)}_{145}\circ \hat T^{(6)}_{246}\circ \hat T^{(6)}_{356}=\hat T^{(6)}_{356}\circ\hat T^{(6)}_{246}\circ\hat T^{(5)}_{145}\circ\hat T^{(5)}_{123}.
\end{gather}
\end{prop}

\begin{proof}
Following  the same procedure as in the proof of Theorem \ref{theo2} and by taking into account the relations (\ref{equiv}), we can prove that $\hat T^{(1)}$ and $\hat T^{(2)}$ are entwining tetrahedron maps. Indeed, in terms of $\hat T^{(1)}$ and $\hat T^{(2)},$  by using (\ref{sube}) to move the transpositions $\tau_{ij}$ to the right and the identity $\tau_{45}\circ \tau_{34}\circ \tau_{25}=\tau_{35}\circ \tau_{24}\circ \tau_{45},$  the entwining tetrahedron equation (\ref{en03}) reads
 \begin{align*}
 \left(S_{13}\circ \bar S_{36}\circ \bar S_{12}\circ S_{15}\circ\bar S_{14}\circ S_{56}\circ\bar S_{26}\circ S_{46}-\bar S_{36}\circ  S_{13}\circ S_{56}\circ\bar S_{26}\circ \bar S_{12}\circ S_{15}\circ  S_{46}\circ \bar S_{14}\right)\circ N=0,
 \end{align*}
where $N:=\tau_{23}\circ\tau_{45}\circ \tau_{24}\circ \tau_{35}.$ Omitting the redundant factor $N,$ we act to the equation above with $S_{36}$ from the left and with $\bar S_{46}$ from the right to obtain
\begin{multline*}
\underbrace{S_{36}\circ S_{13}}_{=S_{13}\circ S_{16} \circ S_{36}}\circ \bar S_{36}\circ \bar S_{12}\circ S_{15}\circ\bar S_{14}\circ S_{56}\circ\bar S_{26}\circ \underbrace{S_{46}\circ \bar S_{46}}_{=id}\\
-\underbrace{S_{36}\circ\bar S_{36}}_{=id}\circ  S_{13}\circ S_{56}\circ\bar S_{26}\circ \bar S_{12}\circ S_{15}\circ  S_{46}\circ\underbrace{\bar S_{14}\circ \bar S_{46}}_{\bar S_{46}\circ \bar S_{16}\circ \bar S_{14}}=0.
 \end{multline*}
Using the assumption that  $S$ is a pentagon  and $\bar S$ a reverse pentagon map, as well as (\ref{equiv}), the equation above simplifies to
\begin{align*}
S_{13}\circ\left(S_{16}\circ \bar S_{12}\circ S_{15}\circ S_{56}\circ \bar S_{26}-S_{56}\circ \bar S_{26}\circ\bar S_{12}\circ S_{15}\circ \bar S_{16}\right)\circ \bar S_{14}=0. 
\end{align*}
In the formula above, by acting to the expression inside the parentheses  with $\bar S_{15}\circ S_{12}\circ\bar S_{16}$ from the left and with $ S_{16}\circ \bar S_{15}\circ S_{12}$  from the right we obtain exactly  the ten term relation (\ref{ten-term}) and that completes the proof.

The proof of the rest of the proposition follows in a similar manner so we omit it. 
\end{proof}

\begin{remark}
So far we considered entwining tetrahedron maps associated with constant pentagon maps that satisfy the ten-term or/and the opposite ten-term relation. These results can be generalized by considering entwining (non-constant) pentagon maps that satisfy the non-constant ten-term (see \cite{Dimakis_Korepanov:2021}) or/and the non-constant opposite ten-term relation.  Then entwining  tetrahedron maps are obtained as well. Indeed, the maps
\begin{align*}
T^{(1)}:=&\bar S_{13}^{(\alpha)}\circ \tau_{23}\circ S_{13}^{(\beta)},& T^{(2)}:=&\bar S_{13}^{(\gamma)}\circ \tau_{23}\circ S_{13}^{(\delta)},&
T^{(3)}:=&\bar S_{13}^{(\epsilon)}\circ \tau_{23}\circ S_{13}^{(\zeta)},& T^{(4)}:=&\bar S_{13}^{(\eta)}\circ \tau_{23}\circ S_{13}^{(\theta)},
\end{align*} 
are entwining tetrahedron maps i.e. they satisfy
\begin{align} \label{tet_22}
   T^{(1)}_{123}\circ  T^{(2)}_{145}\circ T^{(3)}_{246}\circ T^{(4)}_{356}=T^{(4)}_{356}\circ T^{(3)}_{246}\circ T^{(2)}_{145}\circ T^{(1)}_{123},
\end{align}
provided that it holds
\begin{align*}
 S_{12}^{(\beta')}\circ  S_{13}^{(\delta)}\circ  S_{23}^{(\theta)}=&S_{23}^{(\zeta)}\circ S_{12}^{(\beta)}, & (\mbox{entwining pentagon equation}) \\
  \bar S_{23}^{(\eta)}\circ \bar S_{13}^{(\gamma)}\circ \bar S_{12}^{(\alpha')}=&\bar S_{12}^{(\alpha)}\circ \bar S_{23}^{(\epsilon)}, & (\mbox{entwining reverse-pentagon equation}) \\
   \bar S_{13}^{(\alpha')}\circ S_{12}^{(\beta)}\circ\bar S_{14}^{(\gamma)}\circ S_{34}^{(\zeta)}\circ\bar S_{24}^{(\eta)}=&S_{34}^{(\theta)}\circ \bar S_{24}^{(\epsilon)}\circ S_{14}^{(\delta)}\circ\bar S_{13}^{(\alpha)}\circ  S_{12}^{(\beta')}=0,& (\mbox{entwining ten-term relation}).
\end{align*}     

Similarly, the maps
 \begin{align*}
T^{(1)}:=& S_{13}^{(\alpha)}\circ \tau_{12}\circ \bar S_{13}^{(\beta)},& T^{(2)}:=& S_{13}^{(\gamma)}\circ \tau_{12}\circ \bar S_{13}^{(\delta)},&T^{(3)}:=& S_{13}^{(\epsilon)}\circ \tau_{12}\circ \bar S_{13}^{(\zeta)},& T^{(4)}:=& S_{13}^{(\eta)}\circ \tau_{12}\circ \bar S_{13}^{(\theta)},%\\
%T^{(3)}:=& S_{13}^{(\epsilon)}\circ \tau_{12}\circ \bar S_{13}^{(\zeta)},& T^{(4)}:=& S_{13}^{(\eta)}\circ \tau_{12}\circ \bar S_{13}^{(\theta)},
\end{align*} 
satisfy (\ref{tet_22}),
provided that it holds
\begin{align*}
 S_{12}^{(\alpha)}\circ  S_{13}^{(\epsilon)}\circ  S_{23}^{(\eta')}=&S_{23}^{(\eta)}\circ S_{12}^{(\gamma)}, & (\mbox{entwining pentagon equation}) \\
  \bar S_{23}^{(\theta')}\circ \bar S_{13}^{(\zeta)}\circ \bar S_{12}^{(\beta)}=&\bar S_{12}^{(\delta)}\circ \bar S_{23}^{(\theta)}, & (\mbox{entwining reverse-pentagon equation}) \\
   \bar S_{12}^{(\beta)}\circ S_{13}^{(\gamma)}\circ\bar S_{14}^{(\zeta)}\circ S_{24}^{(\eta)}\circ\bar S_{34}^{(\theta')}=&S_{24}^{(\eta')}\circ \bar S_{34}^{(\theta)}\circ S_{14}^{(\epsilon)}\circ\bar S_{12}^{(\delta)}\circ  S_{13}^{(\alpha)}=0,& (\mbox{entwining opposite ten-term relation}).
\end{align*}     
Note that equivalence classes of set-theoretical solutions  of the entwining pentagon and the reverse-pentagon equation, are given in \cite{Kassotakis_2023,CKT}.
% This construction is more restrictive than the one presented in this Section. Nevertheless, interesting results arise hence it requires further studies that %we leave for a future publication.
\end{remark}

\begin{example} \label{ex2}
The pentagon map ${S}_I:\mathbb{CP}^1\times \mathbb{CP}^1 \ni(x,y)\mapsto (u,v)\in \mathbb{CP}^1\times \mathbb{CP}^1$ where 
\begin{align*}
u=&\frac{1}{a}\frac{x}{x+y-ayx}, & v=&x+y-ayx, &a\in \mathbb{CP}^1 \; (\mbox{constant}) 
\end{align*}
 was  introduced in \cite{Kashaev:1999}. Mapping ${S}_I$ together with its inverse $\bar{ {S}_I}:=({S}_I)^{-1}$ satisfy both the ten-term and the variant ten term relations. So according to Theorem \ref{theo2}, mappings $T^{(i)},$ $i=1,\ldots, 4$ satisfy the equations (\ref{en01})-(\ref{enn05}), while the maps $\hat T_i,$ $i=1,\ldots,6$ (see Proposition \ref{propsec3})  the additional entwining equations (\ref{en03})-(\ref{en05}), since we have $\bar{ {S}_I}:=({S}_I)^{-1}$.      
 
 In detail,   mapping $T^{(1)}$  reads
\begin{align}\label{ts1}
  T^{(1)}: (x,y,z)\mapsto &(u,v,w)=\left(\frac{xy}{x+z-azx},x+z-azx,\frac{yz(1-ax)}{x+z-a(y+z)x}\right). % \\
  %T^{(2)}: (x,y,z)\mapsto &\left(\frac{xy}{z+axy-a^2 xyz},z+axy-a^2 xyz,\frac{(1-ax)y}{1-a^2xy}\right),
\end{align}
it satisfies (\ref{en01}) so it is a tetrahedron map. Note that for $a\neq 0,$ $a$ can be scaled out and mapping $T^{(1)}$ is the tetrahedron map (case ($\delta$)) of \cite{Korepanov-1998,Sergeev-1998}. For $a=0$ mapping $T^{(1)}$ reads

\begin{align}\label{hirota}
T^{(1)}_H: (x,y,z)\mapsto& \left(\frac{xy}{x+z},x+z,\frac{yz}{x+z}\right),
\end{align}
and
 corresponds to the  Hirota map \cite{Korepanov-1998,Sergeev-1998}. %Note that this non-abelian version of the Hirota map is  different from the one % introduced in \cite{Kass1} c.f. \cite{Doliwa:2020}. Indeed, one can easily show that  $T^{(1)}_H$ is an involution i.e. $\left(T^{(1)}_H\right)^2=id$, a %feature that is not shared with the non-abelian version of the Hirota map in   \cite{Kass1}. Furthermore, mapping $T^{(1)}_H$ satisfies the following %conservation relations in separated variables
 %\begin{align*}
 %uy^{-1}=&xv^{-1},& wy^{-1}=&zv^{-1},&u^{-1}x=&w^{-1}z,&u+v+w=&x+y+z.&
 %\end{align*}
%Similarly to the commutative case,   $T^{(1)}_H$ is equivalent to any pair of the first three relations above, together with the fourth one. 

We conclude this example by the explicit presentation of the tetrahedron map  $T^{(2)}$ and the entwining tetrahedron maps  $T^{(3)}$ and $T^{(4)}$ associated with ${S}_I$. 
\begin{align*}
T^{(2)}: (x,y,z)\mapsto & \left(\frac{y}{a} \frac{1-a^2xz}{y+z-a z (x+y)},axz,\frac{y+z-a z (x+y)}{1-a^2xz}\right),\\
T^{(3)}: (x,y,z)\mapsto & \left(\frac{xy}{z+axy(1-az)},z+axy(1-az), \frac{(1-ax)y}{1-a^2xy}\right),\\
T^{(4)}: (x,y,z)\mapsto &  \left(\frac{y}{a} \frac{1}{y+z-a y z},ax (y+z-a y z),\frac{(1-ax) (y+z-ayz)}{1-a^2 x(y+z-ayz)}\right).
\end{align*}
The mappings above together with $T^{(1)},$ (see  Theorem \ref{theo2}) satisfy the  equations (\ref{en01})-(\ref{enn05}). Furthermore, mappings $\hat T^{(i)},$ $i=1,\ldots, 6,$  (see proposition (\ref{propsec3})), which are associated with ${S}_I$,  serve as solutions of (\ref{en03})-(\ref{en05}).

%We present the remaining entwining tetrahedon maps $T^{(i)}$ in Appendix \ref{app0}.
%We leave to the interested reader the explicit presentation of the remaining entwining tetrahedron maps $T^{(i)}.$ %and or to apply the theorem above to the pentagon ma  
\end{example}

%\subsection{{\color{blue}From entwining pentagon maps to entwining tetrahedron maps}}
%..
\section{Octorational maps and entwining tetrahedron maps} \label{sec4}

The notion of {\em non-degenerate  rational maps} (quadrirational maps) for the maps $R:\mathbb{X}^2\rightarrow \mathbb{X}^2$ was introduced in \cite{Etingof_1999}. A natural extension of the latter notion for maps $R:\mathbb{X}^n\rightarrow \mathbb{X}^n$ was presented in \cite{2n-rat} under the name {\em $2^n-$rational maps}. In what follows, we present the definition of $2^n-$rational maps when $n=3$, so we have the so-called {\em $2^3-$rational} or   {\em octorational maps}.

\begin{definition}[Octorational maps \cite{2n-rat}]
A  map $R: \mathbb{X} \times \mathbb{X}\times \mathbb{X}\ni(x,y,z)\mapsto (u,v,w) \in \mathbb{X} \times \mathbb{X}\times \mathbb{X}$ is called octorational,
if  the map $R$ and the so called {\em companion maps} (or {\em partial inverses}) $c^1R: (u,y,z)\mapsto (x,v,w),$ $
c^2R: (x,v,z)\mapsto (u,y,w),$ and $
c^3R:(x,y,w)\mapsto (u,v,z),$
are all birational maps.
\end{definition}

\begin{theorem} \label{theo3}
Let $T$ be an octorational tetrahedron map. Then the following entwining tetrahedron equations are  satisfied
\begin{gather}\label{en1}
T^{(1)}_{123}\circ T^{(2)}_{145}\circ T^{(2)}_{246}\circ T^{(2)}_{356}=T^{(2)}_{356}\circ T^{(2)}_{246}\circ T^{(2)}_{145}\circ T^{(1)}_{123},\\ \label{en2}
T^{(4)}_{123}\circ T^{(4)}_{145}\circ T^{(4)}_{246}\circ T^{(1)}_{356}=T^{(1)}_{356}\circ T^{(4)}_{246}\circ T^{(4)}_{145}\circ T^{(4)}_{123},\\ \label{en3}
T^{(3)}_{123}\circ T^{(3)}_{145}\circ T^{(1)}_{246}\circ T^{(4)}_{356}=T^{(4)}_{356}\circ T^{(1)}_{246}\circ T^{(3)}_{145}\circ T^{(3)}_{123},\\ \label{en4}
T^{(2)}_{123}\circ T^{(1)}_{145}\circ T^{(3)}_{246}\circ T^{(3)}_{356}=T^{(3)}_{356}\circ T^{(3)}_{246}\circ T^{(1)}_{145}\circ T^{(2)}_{123},\\ \label{en5}
T^{(4)}_{123}\circ T^{(\bar 3)}_{145}\circ T^{(\bar 3)}_{246}\circ T^{(2)}_{356}=T^{(2)}_{356}\circ T^{(\bar 3)}_{246}\circ T^{(\bar 3)}_{145}\circ T^{(4)}_{123},\\ \label{en6}
T^{(3)}_{123}\circ T^{(\bar 4)}_{145}\circ T^{(2)}_{246}\circ T^{(\bar 3)}_{356}=T^{(\bar 3)}_{356}\circ T^{(2)}_{246}\circ T^{(\bar 4)}_{145}\circ T^{(3)}_{123},\\ \label{en7}
T^{(2)}_{123}\circ T^{(2)}_{145}\circ T^{(\bar 4)}_{246}\circ T^{(\bar 4)}_{356}=T^{(\bar 4)}_{356}\circ T^{(\bar 4)}_{246}\circ T^{(2)}_{145}\circ T^{(2)}_{123},
\end{gather}
as well as the inverse tetrahedron equations\footnote{See Appendix \ref{app1} for their explicit presentation} which are defined by the interchange $i\leftrightarrow \bar i,$ $i=1,\ldots,4$ of the superscripts in (\ref{en1})-(\ref{en7}),
where $T^{(1)}:=T,$ $T^{(\bar 1)}:=T^{-1}$ the inverse of $T,$ $T^{(2)}:=c^1T$ the first companion map of $T$, $T^{(3)}:=c^2T$ the second companion map of $T,$ $T^{(4)}:=c^3T$ the third companion map of $T,$ and $T^{(\bar 2)},$ $T^{(\bar 3)},$  $T^{(\bar 4)}$ their respective inverses.
%The corresponding to the equations (\ref{en1})-(\ref{en7}) tetrahedron maps are genuine.
\end{theorem}
\begin{proof}
The assumption that the mapping $T^{(1)}:(x,y,z)\mapsto \left(u(x,y,z),v(x,y,z),w(x,y,z)\right)$ is~octorational, guarantees the existence of the maps
\begin{gather}\label{oct1}
\begin{aligned}
  T^{(\bar 1)}: \left(u(x,y,z),v(x,y,z),w(x,y,z)\right)\mapsto (x,y,z),\\
   T^{(2)}: \left(u(x,y,z),y,z\right)\mapsto (x,v(x,y,z),w(x,y,z)),  \\
   T^{(\bar 2)}: (x,v(x,y,z),w(x,y,z))\mapsto  \left(u(x,y,z),y,z\right),\\
     T^{(3)}: \left(x,v(x,y,z),z\right)\mapsto (u(x,y,z),y,w(x,y,z)), \\
   T^{(\bar 3)}: (u(x,y,z),y,w(x,y,z))\mapsto \left(x,v(x,y,z),z\right), \\
      T^{(4)}: \left(x,y,w(x,y,z)\right)\mapsto (u(x,y,z),v(x,y,z),z), \\
   T^{(\bar 4)}: (u(x,y,z),v(x,y,z),z)\mapsto \left(x,y,w(x,y,z)\right).
\end{aligned}
\end{gather}
Furthermore, mapping $T^{(1)}$ is a tetrahedron map, that is for a generic point $P_0=(x_1,x_2,x_3,x_4,x_5,x_6)$ it holds.
\begin{align}\label{tt}
T^{(1)}_{123}\left( T^{(1)}_{145}\left( T^{(1)}_{246}\left( T^{(1)}_{356} P_0\right)\right)\right)=T^{(1)}_{356}\left( T^{(1)}_{246}\left( T^{(2)}_{145}\left( T^{(1)}_{123} P_0\right)\right)\right).
\end{align}
To carry on, in order to express the equation above on its components, it is convenient to change the notation. We remove the parentheses on the compositions, we represent the nesting of functions with superscripts and we suppress the variables $x_i$ to $i$. For example, the composed function $u\left(x_1,v\left(x_2,x_4,w\left(x_3,x_5,x_6\right)\right),w\left(x_3,x_5,x_6\right)\right)$ in the new notation reads  $u^{1,v^{2,4,w^{3,5,6}},w^{3,5,6}}$.
Then equation (\ref{tt}) expressed on its components reads
\begin{align}\label{p1}
u^{u^{1,v^{2,4,w^{3,5,6}},v^{3,5,6}},u^{2,4,w^{3,5,6}},u^{3,5,6}}&=u^{u^{1,2,3},4,5},\\ \label{p2}
v^{u^{1,v^{2,4,w^{3,5,6}},v^{3,5,6}},u^{2,4,w^{3,5,6}},u^{3,5,6}}&=u^{v^{1,2,3},v^{u^{1,2,3},4,5},6}\\ \label{p3}
w^{u^{1,v^{2,4,w^{3,5,6}},v^{3,5,6}},u^{2,4,w^{3,5,6}},u^{3,5,6}}&=u^{w^{1,2,3},w^{u^{1,2,3},4,5},w^{v^{1,2,3},v^{u^{1,2,3},4,5},6}},\\ \label{p4}
v^{1,v^{2,4,w^{3,5,6}},v^{3,5,6}}&=v^{v^{1,2,3},v^{u^{1,2,3},4,5},6},\\ \label{p5}
w^{1,v^{2,4,w^{3,5,6}},v^{3,5,6}}&=v^{w^{1,2,3},w^{u^{1,2,3},4,5},w^{v^{1,2,3},v^{u^{1,2,3},4,5},6}},\\ \label{p6}
w^{2,4,w^{3,5,6}}&=w^{w^{1,2,3},w^{u^{1,2,3},4,5},w^{v^{1,2,3},v^{u^{1,2,3},4,5},6}}.
\end{align}
%where we have introduced  a concise notation by removing  the parentheses on the compositions, by representing the nesting of functions with superscripts %and by suppressing the variables $x_i$ to $i.$
%
%(\ref{tetra}) holds, tha in components
%
We now prove that (\ref{en1}) holds, that is
\begin{align}\label{pf1}
T^{(1)}_{123}\left( T^{(2)}_{145}\left( T^{(2)}_{246}\left( T^{(2)}_{356} P_0\right)\right)\right)=T^{(2)}_{356}\left( T^{(2)}_{246}\left( T^{(2)}_{145}\left( T^{(1)}_{123} P_0\right)\right)\right),
\end{align}
for a generic point $P_0=(x_1,x_2,x_3,x_4,x_5,x_6).$
 Since $T$ is octorational, $u$ is a bijection wrt the first argument, $v$ is a bijection wrt the second argument and $w$ is a bijection wrt the third argument. So without loss of generality we may take
 \begin{align*}
 P_0=\left(u\left(x_1,v\left(x_2,x_4,w\left(x_3,x_5,x_6\right)\right),v\left(x_3,x_5,x_6\right)\right),u\left(x_2,x_4,w\left(x_3,x_5,x_6\right)\right),
 u\left(x_3,x_5,x_6\right),x_4,x_5,x_6\right),
 \end{align*}
 %To carry on, it is convenient to change the notation. We remove the parentheses, we represent the nesting of functions with superscripts and we suppress %the variables $x_i$ to $i$. For example, in this notation the point $P_0$ reads
%\begin{align*}
% P_0=\left(u^{1,v^{2,4,w^{3,5,6}},w^{3,5,6}},v^{2,4,w^{3,5,6}},
 %w^{3,5,6},4,5,6\right).
 %\end{align*}
 that in the concise notation reads
 $ P_0=\left(u^{1,v^{2,4,w^{3,5,6}},v^{3,5,6}},u^{2,4,w^{3,5,6}},
 u^{3,5,6},4,5,6\right).$ From the lhs of (\ref{pf1}) we have
 \begin{align*}
   T^{(2)}_{356} P_0=&\left(u^{1,v^{2,4,w^{3,5,6}},v^{3,5,6}},u^{2,4,w^{3,5,6}},3,4,v^{3,5,6},w^{3,5,6}\right):=P_1,\\
 T^{(2)}_{246} P_1=&\left(u^{1,v^{2,4,w^{3,5,6}},v^{3,5,6}},2,3,v^{2,4,w^{3,5,6}},v^{3,5,6},w^{2,4,w^{3,5,6}}\right):=P_2,\\
 T^{(2)}_{145} P_2=&\left(1,2,3,v^{1,v^{2,4,w^{3,5,6}},v^{3,5,6}},w^{1,v^{2,4,w^{3,5,6}},v^{3,5,6}},w^{2,4,w^{3,5,6}}\right):=P_3,\\
 T^{(1)}_{123} P_3=&\left(u^{1,2,3},v^{1,2,3},w^{1,2,3},v^{1,v^{2,4,w^{3,5,6}},v^{3,5,6}},w^{1,v^{2,4,w^{3,5,6}},v^{3,5,6}},w^{2,4,w^{3,5,6}}\right):=P_l.
 \end{align*}
 While from the rhs of (\ref{pf1}) we get
 \begin{align*}
 T^{(1)}_{123} P_0=&\left(u^{K},v^{K},w^{K},4,5,6\right):=P_1',&\mbox{where} &&K:=u^{1,v^{2,4,w^{3,5,6}},v^{3,5,6}},u^{2,4,w^{3,5,6}},u^{3,5,6}.
 \end{align*}
 Due to (\ref{p1}),(\ref{p2}) the point $P_1'$ equivalently reads $P_1'=\left(u^{u^{1,2,3},4,5},u^{v^{1,2,3},v^{u^{1,2,3},4,5},6},w^{K},4,5,6\right).$ So we have
 \begin{align*}
   T^{(2)}_{145} P_1'=& \left(u^{1,2,3},u^{v^{1,2,3},v^{u^{1,2,3},4,5},6},w^{K},v^{u^{1,2,3},4,5},w^{u^{1,2,3},4,5},6\right):=P_2',
 \end{align*}
 \begin{align*}
   T^{(2)}_{246} P_2'=& \left(u^{1,2,3},v^{1,2,3},w^{K},v^{v^{1,2,3},v^{u^{1,2,3},4,5},6},w^{u^{1,2,3},4,5},w^{v^{1,2,3},v^{u^{1,2,3},4,5},6}\right):=P_3',
 \end{align*}
using (\ref{p3}) we have $w^K=u^{w^{1,2,3},w^{u^{1,2,3},4,5},w^{v^{1,2,3},v^{u^{1,2,3},4,5},6}}$ so
\begin{align*}
T^{(2)}_{356} P_3'=&\left(u^{1,2,3},v^{1,2,3},w^{1,2,3},v^{v^{1,2,3},v^{u^{1,2,3},4,5},6},v^\Lambda, w^\Lambda\right):=P_r,
\end{align*}
where $\Lambda:=w^{1,2,3},w^{u^{1,2,3},4,5},w^{v^{1,2,3},v^{u^{1,2,3},4,5},6}.$
Using (\ref{p4})-(\ref{p6}), the points $P_l$ and $P_r$ coincide that serves as a proof that (\ref{en1}) holds.

 The following table presents the generic points $P_0$ that can be used  for the proof of the remaining  entwining tetrahedron equations (\ref{en2})-(\ref{en7}). The proof follows similarly as above so we omit  it.

 \begin{tblr}{c|l} % \hline
      Label of the entwining\\ tetrahedron equation & {}  \\  %[2.5ex]
      \hline
   (\ref{en2}) & $P_0=\left(1,2,w^{1,2,3},4,w^{u^{1,2,3},4,5},w^{v^{1,2,3},v^{u^{1,2,3},4,5},w^{3,5,6}}\right)$  \\  [2.5ex]
   (\ref{en3}) & $P_0=\left(1,v^{1,2,3},3,v^{u^{1,2,3},4,5},5,w^{3,5,6}\right)$ \\ [2.5ex]
   (\ref{en4}) & $P_0=\left(u^{1,2,3},2,3,v^{2,4,w^{3,5,6}},v^{3,5,6},6\right)$ \\ [2.5ex]
   (\ref{en5}) & {$P_0=\left(u^{1,v^{2,4,w^{3,5,6}},v^{3,5,6}},u^{2,4,w^{3,5,6}},u^{w^{1,2,3},w^{u^{1,2,3},4,5},w^{v^{1,2,3},v^{u^{1,2,3},4,5},6}},\right.$\\
   \hspace{7cm}$\left.4, w^{u^{1,2,3},4,5},w^{v^{1,2,3},v^{u^{1,2,3},4,5},6}\right)$} \\ [2.5ex]
   (\ref{en6}) & $P_0=\left(u^{1,v^{2,4,w^{3,5,6}},v^{3,5,6}},u^{v^{1,2,3},v^{u^{1,2,3},4,5},6},u^{3,5,6},v^{u^{1,2,3},4,5},5,w^{3,5,6}\right)$ \\ [2.5ex]
   (\ref{en7}) & $P_0=\left(u^{u^{1,2,3},4,5},u^{2,4,w^{3,5,6}},u^{3,5,6},v^{2,4,w^{3,5,6}},v^{3,5,6},6\right)$ \\
   %\hline
 \end{tblr}
Finally, the proof that (\ref{enn1})-(\ref{enn7}) holds true, follows from item $(2)$ of Remark \ref{rem0}
\end{proof}
The following remarks are in order.
\begin{itemize}
\item In the first four entwining equations (\ref{en1})-(\ref{en4}) and in their associated inverse equations, one of the maps ($T^{(1)}$ and $T^{(\bar 1)}$ in the inverse equations) is a tetrahedron map.
\item In the remaining entwining equations (\ref{en5})-(\ref{en7}) and in their associated inverse equations, in the generic case, none of the maps that participate  is a tetrahedron map.
\item Stronger versions of the Theorem    (\ref{theo3}) can be proven when at least one of the companion maps does not exist. Then the resulting  Theorem reads like Theorem    (\ref{theo3}) but without the entwining relation that the non-existing companion map participates.
\end{itemize}
\begin{example}
The extension of the Hirora map (\ref{hirota0}) on division rings was introduced in \cite{Kass1} c.f. \cite{Doliwa:2020}. Explicitly it reads
\begin{align} \label{hirotaq}
T_{H}:\mathcal{A}\times\mathcal{A}\times\mathcal{A}\ni(x,y,z)\rightarrow (u,v,w)\in \mathcal{A}\times\mathcal{A}\times\mathcal{A},
\end{align}
where
\begin{align*}
  u=&(x+z)^{-1}xy,&v=&x+z,&w=&(x+z)^{-1}zy.
\end{align*}
with $\mathcal{A}$ a division ring. This non-abelian version of the Hirota map is invertible with inverse that reads
\begin{align*}
(T_{H})^{-1}:(x,y,z)\mapsto \left(yx(x+z)^{-1},x+z,yz(x+z)^{-1}\right).
\end{align*}
Furthermore, it admits a first companion map $ c^1T_H$ and its inverse,  a third companion map $ c^3T_H$ and their respective inverses, while its second companion $c^2T_H$  does not exist.
Explicitly the companion maps of the non-abelian Hirota map and their inverses   read:
\begin{align*}
 c^1T_H:(x,y,z)\mapsto & \left(zx(y-x)^{-1},zy(y-x)^{-1},y-x\right),\\
\left(c^1T_H\right)^{-1}:(x,y,z)\mapsto &\left((y-x)^{-1}xz,(y-x)^{-1}yz,y-x\right), \\
 c^3T_H:(x,y,z)\mapsto &\left(y-z,xy(y-z)^{-1},xz(y-z)^{-1}\right), \\
  \left(c^3T_H\right)^{-1}:(x,y,z)\mapsto &\left(y-z,(y-z)^{-1}yx,(y-z)^{-1}zx\right).
\end{align*}

% Also it holds 
%\begin{align*}
%\left(T^{(1)}_H\right)^2=\left(c^1T^{(1)}_H\right)^2=\left(c^3T^{(1)}_H\right)^2=id,
%\end{align*}
%so $\left(T^{(1)}_H\right)^{-1}\equiv T^{(1)}_H,$ $\left(c^iT^{(1)}_H\right)^{-1}\equiv c^iT^{(1)}_H,$ $i=1,3$. 

Following the notation of the Theorem \ref{theo3} there is 
\begin{align*}
T^{(1)}:=&T_H,&T^{(\bar 1)}:=&(T_H)^{-1},&T^{(2)}:=&c^1T_H,&T^{(\bar 2)}:=&(c^1T_H)^{-1},&T^{(4)}:=&c^3T_H,&T^{(\bar 4)}:=&(c^3T_H)^{-1},
\end{align*}
which satisfy (\ref{en1}),(\ref{en2}),(\ref{en7}),(\ref{enn1}),(\ref{enn2}) and (\ref{enn7}) entwining tetrahedron equations.

For examples of tetrahedron maps with the full octorational symmetry we refer to \cite{Dimakis:2019,Rizos_2020}. It is easy to verify that these tetrahedron maps their inverses, their  associated three companions maps together with their inverses, all exist and differ. So according to Theorem \ref{theo3}, they satisfy all fourteen entwining tetrahedron equations (\ref{en1})-(\ref{en7}), (\ref{enn1})-(\ref{enn7}).%  We leave the explicit presentation of their companion maps to the interested reader.
\end{example}

\section*{Acknowledgements}
\parbox{.135\textwidth}{\begin{tikzpicture}[scale=.03]
\fill[fill={rgb,255:red,0;green,51;blue,153}] (-27,-18) rectangle (27,18);
\pgfmathsetmacro\inr{tan(36)/cos(18)}
\foreach \i in {0,1,...,11} {
\begin{scope}[shift={(30*\i:12)}]
\fill[fill={rgb,255:red,255;green,204;blue,0}] (90:2)
\foreach \x in {0,1,...,4} { -- (90+72*\x:2) -- (126+72*\x:\inr) };
\end{scope}}
\end{tikzpicture}} \parbox{.85\textwidth}
{This research is part of the project No. 2022/45/P/ST1/03998  co-funded by the National Science Centre and the European Union Framework Programme
 for Research and Innovation Horizon 2020 under the Marie Sklodowska-Curie grant agreement No. 945339. For the purpose of Open Access, the author has applied a CC-BY public copyright licence to any Author Accepted Manuscript (AAM) version arising from this submission.}

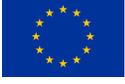
\appendix

%\section{The entwining tetrahedron maps associated with mapping $\mathfrak{S}_I$ and its inverse}\label{app0}

%The entwining tetrahedron maps $T^{(i)},$ $i=1,\ldots 10$ obtained by applying Theorem \ref{theo2} to the pentagon map $\mathfrak{S}_I$ and its inverse %$\bar{\mathfrak{S}_I}$ respectively read
%\begin{align*}
%T^{(1)}: (x,y,z)\mapsto &\left(x(x+z-azx)^{-1}y,x+z-azx,z(1-ax)(x+z-a(y+z)x)^{-1}y\right),\\
%T^{(2)}: (x,y,z)\mapsto &
%\end{align*}

\section{Inverse entwining tetrahedron equations}\label{app1}
The entwining tetrahedron equations (\ref{en1})-(\ref{en7}) are accompanied with the inverse entwining tetrahedron equations which are obtained by the interchange $\i\leftrightarrow \bar i,$ $i=1,\ldots,4$ of the superscripts on (\ref{en1})-(\ref{en7}). Explicitly they read
\begin{gather}\label{enn1}
T^{(\bar 1)}_{123}\circ T^{(\bar 2)}_{145}\circ T^{(\bar 2)}_{246}\circ T^{(\bar 2)}_{356}=T^{(\bar 2)}_{356}\circ T^{(\bar 2)}_{246}\circ T^{(\bar 2)}_{145}\circ T^{(\bar 1)}_{123},\\ \label{enn2}
T^{(\bar 4)}_{123}\circ T^{(\bar 4)}_{145}\circ T^{(\bar 4)}_{246}\circ T^{(\bar 1)}_{356}=T^{(\bar 1)}_{356}\circ T^{(\bar 4)}_{246}\circ T^{(\bar 4)}_{145}\circ T^{(\bar 4)}_{123},\\ \label{enn5}
T^{(\bar 3)}_{123}\circ T^{(\bar 3)}_{145}\circ T^{(\bar 1)}_{246}\circ T^{(\bar 4)}_{356}=T^{(\bar 4)}_{356}\circ T^{(\bar 1)}_{246}\circ T^{(\bar 3)}_{145}\circ T^{(\bar 3)}_{123},\\ \label{enn6}
T^{(\bar 2)}_{123}\circ T^{(\bar 1)}_{145}\circ T^{(\bar 3)}_{246}\circ T^{(\bar 3)}_{356}=T^{(\bar 3)}_{356}\circ T^{(\bar 3)}_{246}\circ T^{(\bar 1)}_{145}\circ T^{(\bar 2)}_{123},\\ \label{enn3}
T^{(\bar 4)}_{123}\circ T^{(3)}_{145}\circ T^{(3)}_{246}\circ T^{(\bar 2)}_{356}=T^{(\bar 2)}_{356}\circ T^{(3)}_{246}\circ T^{(3)}_{145}\circ T^{(\bar 4)}_{123},\\ \label{enn4}
T^{(\bar 3)}_{123}\circ T^{(4)}_{145}\circ T^{(\bar 2)}_{246}\circ T^{(3)}_{356}=T^{(3)}_{356}\circ T^{(\bar 2)}_{246}\circ T^{(4)}_{145}\circ T^{(\bar 3)}_{123},\\  \label{enn7}
T^{(\bar 2)}_{123}\circ T^{(\bar 2)}_{145}\circ T^{(4)}_{246}\circ T^{(4)}_{356}=T^{(4)}_{356}\circ T^{(4)}_{246}\circ T^{(\bar 2)}_{145}\circ T^{(\bar 2)}_{123}.
\end{gather}

%\bibliographystyle{unsrt}
%\bibliography{ref-tet}

\end{document}